\documentclass{ec-acmsmall-tweaked}
\usepackage[numbers]{natbib}
\usepackage[utf8]{inputenc}

\usepackage{amsmath,amssymb,amsfonts,latexsym} 
\usepackage{bbm} 
\usepackage{url, algorithm, algorithmic, verbatim}
\usepackage{etoolbox,xparse,xspace}

\usepackage{tikz,pgfplots}
\usetikzlibrary{arrows,shapes,positioning,calc}

\newcommand{\Comments}{0}
\definecolor{gray}{gray}{0.5}
\definecolor{darkgreen}{rgb}{0,0.5,0}
\newcommand{\mynote}[2]{\ifnum\Comments=1\textcolor{#1}{#2}\fi}
\newcommand{\ian}[1]{\mynote{blue}{[IAK: #1]}}
\newcommand{\raf}[1]{\mynote{darkgreen}{[RMF: #1]}}

\renewcommand{\O}{\mathcal{O}}

\renewcommand{\vec}[1]{{\mathbf{#1}}}

\newcommand{\x}{\vec{x}}

\newcommand{\conv}{\convhull}

\newcommand{\clo}{\mathrm{cl}}

\newcommand{\cell}{\mathrm{cell}}

\def\reals{\mathbb{R}}

\def\extreals{\mathbb{\overline{R}}}

\newcommand{\argmax}{\mathop{\mathrm{argmax}}}

\renewcommand{\conv}{\mathrm{conv}}

\newcommand{\first}{daskalakis2013mechanism}
\newcommand{\second}{giannakopoulos2014duality}
\newcommand{\third}{daskalakis2015strong}

\renewcommand{\Comments}{0} 

\begin{document}

\title{Optimal Auctions with Restricted Allocations}
\author{
IAN A. KASH
\affil{Microsoft Research}
RAFAEL M. FRONGILLO
\affil{University of Colorado Boulder}
}

\begin{abstract}
We study the problem of designing optimal auctions under restrictions on the set of permissible allocations.  In addition to allowing us to restrict to deterministic mechanisms, we can also indirectly model non-additive valuations.  We prove a strong duality result, extending a result due to \citet{\third}, that guarantees the existence of a certificate of optimality for optimal restricted mechanisms.  As a corollary of our result, we provide a new characterization of the set of allocations that the optimal mechanism may actually use.  To illustrate our result we find and certify optimal mechanisms for four settings where previous frameworks do not apply, and provide new economic intuition about some of the tools that have previously been used to find optimal mechanisms.
\end{abstract}
\maketitle

\section{Introduction}

The problem of revenue maximization a monopolist seller with a single item to sell to a single buyer was solved by \citet{myerson1981optimal}, who showed that the optimal mechanism provides a menu with two options to the buyer: get the item with probability zero for free or get the item with probability one at a price that depends on the distribution of the buyer's value.  The setting with multiple items is significantly more challenging, however, even when restricted to two items, as a line of work from economics has shown~\citep{mcafee1988multidimensional,rochet1998ironing,thanassoulis2004haggling,manelli2006bundling,manelli2007multidimensional,chung2007non-differentiable,pavlov2011optimal,noldeke2015implementation}.

Significant progress has been made in the past few years, using a new duality technique introduced by \citet{\first}.  Subsequent work has refined this technique, allowing for the design of both optimal and approximately optimal mechanisms in a variety of settings~\citep{\second,\third,giannakopoulos2015bounding,giannakopoulos2015selling}.  See \citet{daskalakis2015multi-item} for a survey.
However, all this previous work has had two key restrictions:
\begin{enumerate}
\item valuations are additive, in the sense that the value for a bundle of items is the sum of values for the individual items, and
\item the mechanism can choose any (randomized) allocation of the items.
\end{enumerate}

We relax this second restriction, by showing how to extend this duality approach to mechanism design problems where the set of permissible allocations is more general.  This allows problems with multiple copies of items or a requirement that the mechanism be deterministic.  Indirectly, it also allows us to relax the first restriction.  The reason is that {\em any} problem can be made additive by redefining the space of types to assign a value to each possible bundle.  However, after this transformation the mechanism must be restricted to only offer probability distributions over bundles.  Thus, by allowing restrictions on allocations we also enlarge the space of valuations to which these techniques can be applied.

\subsection{General Duality}
Our main result is a generalization of the strong duality framework of \citet{\third} between optimal mechanism design and a dual problem that can be interpreted as finding an optimal matching, and we broadly follow their approach.
\begin{itemize}
\item The standard way of phrasing the optimal mechanism design problem is in terms selecting a convex consumer surplus function $u$.  The constraint that $u$ is a valid mechanism corresponds requiring that the subgradients of $u$ correspond to permissible allocations.
The revenue of a given $u$ is then easily written as an expression in terms of both $u$ and its subgradients.  However, this formulation seems difficult to work with, so Daskalakis et al. use the divergence theorem to constuct a measure $\mu$ such that the objective function can be rewritten as $\int u d\mu$, thus excluding the subgradients of $u$ from the objective function.  We adopt this approach wholesale.
\item To define feasible dual solutions, their approach makes use of the notion of convex dominance.  A measure $\mu$ convex dominates a measure $\mu'$ if $\int u d\mu \geq \in u d\mu'$ for all (non-decreasing) convex $u$.  This makes convex dominance an instance of an integral stochastic order~\citep{muller2002comparison}.  The key to our approach is to use more general such orders, in particular taking the order induced by set of $u$ that are actually feasible mechanisms rather than all convex $u$.
\item Using our more general stochastic order, we prove a duality theorem of which their result is a special case.
\item To illustrate our framework, we illustrate optimal auctions and corresponding dual witnesses for a number of settings that prior techniques could not handle.  First, we restrictions on the set of possible allocations, specifically the case of two item i.i.d uniform [0,1] items with the restriction to allocate at most one item and the restriction to allocate exactly one item.  Second, we consider two i.i.d exponentially distributed items with the restriction to deterministic mechanisms.  The previous examples have restricted the set of allocations but used additive valuations.  Our final example looks at a setting with non-additive valuations (there are two items and the goods are complements) and unrestricted allocations, and shows that bundling is optimal.
\end{itemize}

\subsection{ Finding Good Mechanisms}
 Our main theorem is a tool that certifies that a mechanism is optimal---it does not directly provide guidance on how to find the optimal mechanism in the first place.  Thus, we provide several additional results that provide additional structure to aid in the search for good mechanisms.  
\begin{itemize}
\item We prove a lemma showing that optimal mechanisms only use allocations that are exposed points of the set of allowable allocations.  In particular, all previous known optimal mechanisms for the additive case are zero set mechanisms~\citep{\third}, where the buyer is either allocated nothing or at least one item with probability 1.  We show that in fact a slightly weaker version of this is a necessary condition, and that suitably generalized it holds beyond the additive setting.
\item One useful observation made by \citet{\third} for the case where the optimal mechanism has a finite number of allocations is that if $\cell(s)$ is the set of types that receive allocation $s$ then $\int_{cell(s)} \mu=0$.  We provide economic intuition for this fact, showing that it corresponds to the simple first-order condition that an optimal mechanism prices each allocation optimally.  That is, the derivative of revenue with respect to changing the price of any one allocation is 0.
\item Any feasible dual solution provides an upper bound on the revenue of the optimal mechanism. \citet{giannakopoulos2015bounding} uses this observation to certify the approximate optimality of mechanisms.  We examine our dual problem, and show that it suggests that a better understanding of the integral stochastic orders we use would shed light on the question of how well the optimal mechanism can be approximated by a deterministic mechanism.
\end{itemize}

\section{Preliminaries}

Let $\O$ be a set of outcomes.  There is a single agent whose preference over outcomes can be represented as a function $x \in \reals^\O_{\geq 0}$, known as the agent's {\em type}.  That is, this function associates each outcome with a non-negative value for the agent.
We work in a Bayesian setting, so there is a distribution $f$ over $\reals^\O_{\geq 0}$ from which the agent's type is drawn.
An allocation $s$ is a linear function in the dual space $\reals^\O \rightarrow \reals$.  The value of the agent is then $s(x)$.  If $\O$ is finite, both $x$ and $s$ can be represented as vectors of dimension $|\O| = n$ and $s(x) = s \cdot x$.  Because the techniques we use to derive the objective function (see Section~\ref{sec:objective}) are restricted to the finite case we assume that $\O$ is finite and use the $s \cdot x$ notation, but most of our remaining techniques do not rely on this assumption and would generalize if this analysis does.

Despite calling $\O$ a set of outcomes, our formalism does not restrict $s$ to ensure that at most one outcome happens or that a given outcome can happen a single time.  Thus our formalism can handle auctions with multiple items or even multiple copies of a single item with additive valuations in two distinct ways.  $\O$ can be the set of individual items (or types of items), in which case $s$ indicates how many (in expectation) of each item are allocated.  Alternatively, we can let $\O$ be the set of all bundles of items, and $s$ give a probability distribution over bundles. \raf{At first I was confused what the alternative would be; might be good to clarify.} When there is a choice of representations, it is convenient to select one that results in a type distribution $f$ that meets the requirements of the derivation in Section~\ref{sec:objective} while also minimizing $|\O|$.  For an example that takes advantage of this flexibility, see Section~\ref{sec:example-nonadd}.

As this possibility of multiple representations suggests, it may be natural to have a restricted set of types $X$ and feasible allocations $S$ to accurately describe a given problem.  For example, we can restrict $X$ to capture additive valuations or that goods are substitutes, while we can restrict $S$ to impose constraints such as limited quantities of an item available, that $s$ is a probability distribution, or that the mechanism is deterministic (i.e. $s$ has an integer value in each coordinate).  Previous work on optimal mechanisms for the additive case has considered various restrictions on $X$, such as requiring values to lie in $[0,1]^{|\O|}$, but has always taken $S$ to be $[0,1]^{|\O|}$.  Without loss of generality, we require $X$ to be convex.

\subsection{Convex Representation}

In this setting it is without loss of generality to restrict to {\em truthful direct revelation mechanisms}.  These consist of a rule that takes a report of $x$ from the agent and then selects an allocation $s$ and payment $p$, which gives the agent a utility of $u(x) = s(x) \cdot x - p(x)$.  A standard fact is that we can represent such mechanisms as convex functions $u$ where $s(x)$ is a subgradient of $u$ at $x$ and $p(x) = u(x) - s(x) \cdot x$.
Thus, the set $U(X,S)$ of convex functions on
$X$
that are subdifferentiable with subgradients in $S$\footnote{Of course, any $u$ with subgradients in $S$ in general has subgradients in $\conv(S)$ as well.  We say ``subgradients in $S$'' to mean that $u$ is supported by $S$ alone: we can write $u(x) = \sup_{s\in S} s \cdot x + c_s$ for some constants $c_s$.} corresponds exactly to the truthful direct revelation mechanisms.
We further require that the mechanism be individually rational, a restriction we enforce via the objective function described in the next section.

Given a set of functions $U$, a geometric structure that will prove useful for our arguments is the cone generated by $U$, which consists of all scaled versions of elements of $U$.  That is, $U^\circ = \{ \lambda u ~|~ \lambda \geq 0 \wedge u \in U \}$.  Our arguments also make use of the convex hull $\conv(U)$ and the closure $\clo(U)$.

Finally, we make use of the fact that since $u$ is restricted to have subgradients in $S$, a simple upper bound on the difference in the value of $u$ between points $x$ and $y$ is $u(x) - u(y) \leq \sup_{s \in S} s \cdot (x-y) = \ell_S(x,y)$.  We will be using the concise $\ell_S$ notation for this supremum frequently.
We require $S$ to be closed and bounded, which ensures that this supremum is finite and realized.
To see the importance of closure, consider the case of selling a single item, with $S = [0,1)$.  There is no optimal mechanism restricted to this $S$.  We can approximate the optimal auction with $S = [0,1]$ arbitrarily well, but cannot realize it because the optimal auction only offers the item for sale with probability 1.  Similarly, without boundedness we could take $S = [0,\infty)$ and there would be no optimal auction because it is optimal to sell as much of the item as possible.

\subsection{Objective Function}
\label{sec:objective}

Following \citet{\third}, as long as the pdf of the type distribution $f$ is differentiable, we can write our problem as trying to choose the optimal consumer surplus function $u$ with respect to a measure $\mu$.  That is, we make use of the following definition and result from their work.

\begin{definition}
Let $X' \subset  \reals^n_{\geq 0}$ be a well-behaved%
\footnote{$X'$ is well behaved if it is a Jordan-measurable bounded Lipschitz domain} set of types.
Let $f : X' \rightarrow \reals_{\geq 0}$ be a (differentiable) probability density function with bounded partial derivatives.
Let $z_0 \in \reals^n_{\geq 0}$ be any point which is coordinate-wise less than all points in $X'$.
The signed Radon measure $\mu$ is a {\em transformed measure} of $f$ if the relation
$$\int h d\mu = h(z_0) + \int_{\partial X'} h(z)f(z) z \cdot \hat{n} dz - \int_{X'} h(z)(\nabla f(z) \cdot z + (n+1)f(z))dz$$
holds for all continuous bounded functions $h : \reals^n \rightarrow \reals$, where $\hat{n}$ denotes the outer unit normal field to the boundary $\partial X'$.
\end{definition}

\begin{theorem}[\citep{daskalakis2015strong}]
\label{thm:mu}
Let $X' \subset  \reals^n_{\geq 0}$ be a well-behaved set of types.
Let $f : X' \rightarrow \reals_{\geq 0}$ be a (differentiable) probability density function with bounded partial derivatives.
Then the problem of determining the optimal IC and IR mechanism for a single additive buyer whose type is distributed according to $f$ is equivalent to solving the optimization problem
$$\sup_{u \in U(X,S)} \int_X u d\mu$$
Where $X = [0,M]^n \supset X'$ and $\mu$ is a transformed measure of $f$. \raf{Clarify where $M$ comes from.}
\end{theorem}

Essentially, their result says that as long as we pick a ``nice'' $f$, it has a transformed measure $\mu$ such that finding the optimal mechanism is the same as maximizing $\int_X u d\mu$, where we have enlarged the set of types to be a hypercube.  While this theorem requires $X'$ to be bounded, they show that this can be relaxed as long as $f$ decays sufficiently rapidly, and we use this extension in Section~\ref{sec:example-deterministic}.  Earlier work from has used more restricted versions of this observation.
For example, in economics, \citet{manelli2006bundling} and \citet{thanassoulis2004haggling} used integration by parts to get such a measure $\mu$,
while in computer science, special cases of this approach were used by \citet{\first} and \citet{\second}.

\subsection{Measure Theory}

As our new objective involves a measure $\mu$, we introduce some notation and definitions from measure theory.  Let $X \subset \reals^n$ be compact.  Then
\begin{itemize}
\item $\mu \in Radon(X)$, the set of signed Radon measures supported within $X$.
\item By the Jordan decomposition theorem, $\mu = \mu_+ - \mu_-$ where $\mu_+,\mu_- \in Radon_+(X)$, the set of positive Radon measures supported within $X$.  Intuitively $\mu_+$ is the positive part of $\mu$ and $\mu_-$ is the negative part.
\item For $\gamma \in Radon(X \times X)$ (or $Radon_+(X \times X)$), the marginal measures $\gamma_1$ and $\gamma_2 \in Radon(X)$ (or $Radon_+(X)$) are defined as $\gamma_1(X') = \gamma(X',X)$ and $\gamma_2(X') = \gamma(X,X')$.
\end{itemize}

\subsection{Stochastic Orders}

Having reduced our objective function to something of the form $\int_X u d\mu$, previous duality based approaches have observed that it is useful to upper bound this with a transformation $\mu'$ of $\mu$, but one that guarantees that $\int u d\mu \leq \int u d\mu'$.  The problem is that we do not know what $u$ is, so we need this relation to hold for a large enough class of functions.  This guarantee is provided by integral stochastic orders~\citep{muller2002comparison}.  Let $U$ be some set of functions $X \to \reals$.  Then $\mu' \succeq_{U} \mu$ if $\int_X u d\mu' \geq \int_X u d\mu$ for all $u \in U$.  Orders used in prior work include the stochastic order ($U$ is all increasing functions) and the convex order ($U$ is all convex functions).  The natural order for our setting is to take $U$ to be all feasible candidates for $u$.  That is, $U = U(X,S)$.  Since this order is entirely determined by $X$ and $S$, for brevity we will write $\succeq_{X,S}$ rather than $\succeq_{U(X,S)}$.

\section{General Strong Duality}

Our main result is the following theorem, which establishes strong duality between the optimal mechanism $u$ an a matching $\gamma$.

\begin{theorem}
\label{thm:strong}
Let $X = [0,M]^n$ and let $\mu \in Radon(X)$ such that $\mu(X) = 0$.  Then
\begin{equation}
\label{eqn:strong}
\sup_{u \in U(X,S))} \int_X ud\mu = \min_{\substack{\gamma \in Radon_+(X \times X)\\ \gamma_1 - \gamma_2 \succeq_{X,S} \mu}} \int_{X \times X} \ell_S(x,y)d\gamma(x,y)
\end{equation}
\end{theorem}

Note that it is part of the theorem that the $\min$ is realized.  While we do not prove that the $\sup$ is realized, it is realized whenever an optimal mechanism exists.  Theorem 2 of \citet{\third} is a special case of theorem with $S= [0,1]^n$.  They directly prove that the $\sup$ is realized in this case.

The structure of our proof parallels that of \citet{\third}, which in turn parallels the structure of the proof of Monge-Kantorovich duality for optimal transport given by~\citet{villani2003topics}.  Just as the technical details of their proof differ from the proof of Monge-Kantorovich duality due to an added convexity constraint, so to do ours differ from theirs due to an added subgradient constraint.

To prove the theorem, we use three lemmas.  We begin with a simple proof of weak duality.  Note that weak duality alone is enough to apply our framework to practical examples.  The remainder of our results are only needed to show that in fact our framework always applies.
\begin{lemma}
\label{lem:weak}
\[\sup_{u \in U(X,S)} \int_X ud\mu \leq \inf_{\substack{\gamma \in Radon_+(X \times X)\\ \gamma_1 - \gamma_2 \succeq_{X,S} \mu}} \int_{X \times X} \ell_S(x,y)d\gamma(x,y)~.\]
\end{lemma}

\begin{proof}
Fix feasible $u$ and $\gamma$.
\begin{align*}
\int_X u d\mu &\leq \int_X u d(\gamma_1-\gamma_2)\\
&= \int_{X \times X} u(x) - u(y) d\gamma(x,y)\\
&\leq \int_{X \times X}\ell_S(x,y)d\gamma(x,y)~.
\end{align*}
The first inequality follows by domination, the second by $\ell$ being an upper bound on the difference in the value of $u$ between any two points.
\end{proof}

The next lemma is the most technical, and shows strong duality for a related problem.  Ultimately we will combine this with weak duality to show strong duality for our desired problem.
\begin{lemma}
\label{lem:main}
$$\sup_{\substack{\phi(x)-\psi(y) \leq \ell_S(x,y)\\\phi,\psi \in \clo(conv(U^\circ(X,S)))}} \int_X \phi d\mu_+ - \int_X \psi d\mu_- = \min_{\substack{\gamma \in Radon_+(X \times X)\\ \gamma_1 \succeq_{X,S} \mu_+\\ \mu_- \succeq_{X,S} \gamma_2}} \int_{X \times X} \ell_S(x,y)d\gamma(x,y)$$
\end{lemma}

\begin{proof}
The proof makes use of Fenchel-Rockafellar Duality (specifically in its symmetric form~\cite[Exercise 4.4.17]{borwein2010convex}, although we use the version of the conditions for strong duality given by~\cite{villani2003topics}).
\raf{This theorem being here makes the proof hard to parse, and gets me every time, so I \texttt{quote}d it}
\begin{quote}
  \noindent\textsc{Theorem (Fenchel-Rockafellar Duality).}
    Let $E$ be a Banach space, $E^*$ its dual, and proper $\Theta,\Xi:
    E \to\extreals$ with $\Theta$ convex and $\Xi$ concave.  Then
$$ \sup (\Xi - \Theta) \leq \inf (\Theta^* - \Xi^*),$$
where $\Theta^*$ and $\Xi^*$ are the convex and concave duals
respectively.  Further, if there exists $e \in E$ where $\Theta(e)$
and $\Xi(e)$ are finite and $\Theta$ is continuous at $e$, this holds
with equality and the infimum is realized.
\end{quote}
Our proof also makes use of the Fenchel-Moreau theorem, which says that a proper convex $f = f^{**}$ iff $f$ is lsc (equivalently closed).

Let $CB(X \times X)$ be the set of continuous, bounded functions $X \times X \rightarrow \mathbb{R}$.   Note its dual space is $Radon(X \times X)$.  To make use of the theorem, we can split the domain constraint into two pieces and phrase the optimization problem on the right hand side as
$\inf_\gamma \Theta^*(\gamma) - \Xi^*(\gamma)$, where

$$\Theta^*(\gamma) = \begin{cases}\int_{X \times X} \ell_S(x,y)d\gamma(x,y) &\mbox{if } \gamma \in Radon_+(X \times X)\\ +\infty &\mbox{otherwise}\end{cases}$$
and
$$\Xi^*(\gamma) = \begin{cases} 0 &\mbox{if } \gamma_1 \succeq_{X,S} \mu_+ \mbox{ and } \mu_- \succeq_{X,S} \gamma_2 \\ -\infty &\mbox{otherwise}\end{cases}$$

Note that we use a $*$ because these functions are on the dual space.  $\Theta^*$ is closed and convex as it is linear on its domain, which is a closed convex set.  Similarly, $\Xi^*$ is closed and concave as its domain is a closed convex set.  Thus, by Fenchel-Moreau we have that $\Theta^* = \Theta^{***}$ and $\Xi^* = \Xi^{***}$.

We can calculate $\Theta^{**}$ as
\begin{align*}
\Theta^{**}(f) &= \sup_{\gamma \in Radon_+(X \times X)} \int_{X \times X} (f(x,y) - \ell_S(x,y))d\gamma(x,y)\\
&= \begin{cases} 0 &\mbox{if } \ell_S \geq f\\ +\infty &\mbox{otherwise}\end{cases}
\end{align*}
For the second equality, note that in the first case the difference is non-positive at every point and $\gamma$ is non-negative, so the the supremum is realized by taking $\gamma = 0$.  In the second case, take any point where the difference is positive and choose a $\gamma$ that puts an arbitrarily high measure on that point and 0 measure elsewhere.

Note that $(\Theta^{**})^*$ = $\Theta^*$ and $\Theta^{**}$  is closed and convex.  Thus we can take $\Theta = \Theta^{**}$ and its dual and double dual will be $\Theta^*$ and $\Theta^{**}$ respectively.

For $\Xi^{**}$, we proceed by guessing a closed concave $\Xi^{**}$ and confirming that its dual is in fact $\Xi^*$.  Take $$\Xi^{**}(f) = \begin{cases} \int_X \phi d(\mu_+) - \int_X \psi d(\mu_-)&\mbox{if } \exists \phi,\psi\in \clo(\conv(U^\circ(X,S))):\; f(x,y) = \phi(x) - \psi(y)\\ -\infty &\mbox{otherwise,}\end{cases}$$ 
\begin{align*}
(\Xi^{**})^*(\gamma) &= \inf_{f \in CB{X \times X}} \int_{X \times X} f(x,y)d(\gamma(x,y)) - \Xi^{**}(f)\\
&= \inf_{\phi,\psi \in \clo(\conv(U^\circ(X,S)))} \int_X \phi d(\gamma_1 - \mu_+) + \int_X \psi d(\mu_- - \gamma_2)\\
&= \begin{cases} 0 &\mbox{if } \gamma_1 \succeq_{X,S} \mu_+ \mbox{ and } \mu_- \succeq_{X,S} \gamma_2\\ -\infty &\mbox{otherwise}\end{cases}\\
&= \Xi^*(\gamma)
\end{align*}
The penultimate equality follows from the following observation.  If $\gamma_1\succeq_{X,S} \mu_+$, the integral is always non-negative.  Taking $\phi = 0$ makes it 0, and similarly for $\phi$.  If $\gamma_1 \not\succeq_{X,S} \mu_+$, then there is some $\phi \in U(X,S)$ for which the integral is negative.  As we are allowed to scale it arbitrarily in $U^\circ(X,S)$, the infimum is infinite (and similarly for $\psi)$.
Since $\Xi^{**}$ is closed and concave, we can again take $\Xi = \Xi^{**}$.

The final detail to apply the theorem is to find an $f$ such that $\Theta(f)$ and $\Xi(f)$ are finite and $\Theta$ is continuous at $f$.
Take $f(x,y) = - 1$.  $\Theta(f) = 0$ and $\Xi(f) = - \int_X d\mu_-$.  $\Theta(f') = 0$ for all $f'$ such that $||f-f'||_\infty \leq 1$.  Thus, $\Theta$ is continuous at $f$.

We now apply Fenchel-Rockafellar duality, giving us the following:
\begin{align*}
\sup_{f \in CB(X \times X)} \Xi(f) - \Theta(f) &= \min_{\gamma \in Radon(X \times X)} \Theta^*(\gamma) - \Xi^*(\gamma)\\
\sup_{\substack{\phi(x)-\psi(y) \leq \ell_S(x,y)\\\phi,\psi \in \clo(\conv(U^\circ(X,S)))}} \int_X \phi d\mu_+ - \int_X \psi d\mu_- &= \min_{\substack{\gamma \in Radon_+(X \times X)\\ \gamma_1 \succeq_{X,S} \mu_+\\ \mu_- \succeq_{X,S} \gamma_2}} \int_{X \times X} \ell_S(x,y)d\gamma(x,y)
\end{align*}
\end{proof}

The final lemma takes the supremum over two functions for which we proved strong duality in the previous lemma, and shows that we can convert that not just a single function, but one in $U(X,S)$ as well.
\begin{lemma}
\label{lem:21}
\begin{equation}
\label{eqn:21}
\sup_{\substack{\phi(x)-\psi(y) \leq \ell_S(x,y)\\\phi,\psi \in \clo(\conv(U^\circ(X,S)))}} \int_X \phi d\mu_+ - \int_X \psi d\mu_-= \sup_{u \in U(X,S)} \int_X ud\mu
\end{equation}
\end{lemma}

\begin{proof}
First, observe that for any $u$ feasible on the right hand side of \eqref{eqn:21}, $\phi = \psi = u$ is feasible on the left hand side with the same value.  Thus, the left hand side is at least the right hand side.  

Now consider $\phi$ and $\psi$ feasible for the left hand side.  By the constraints, $\psi(y) \geq \phi(x) - \ell_S(x,y)$ for all $x \in X$.  Thus, define $\bar\psi(y) = \sup_x \phi(x) - \ell_S(x,y)$.  This is finite on $X$ as $\bar\psi(y) \leq \psi(y)$.  We claim that $\bar\psi \in U(X,S)$.  To show this, it suffices to find a subgradient $s_y \in S$ of $\bar\psi$ at each $y \in X$.  Observe that since $X$ is compact, $S$ is closed, and $\phi$ is continuous, we can write $\bar\psi(y) = \max_{x \in X, s \in S} \phi(x) - s \cdot (x-y)$.  Let $s_y$ and $x_y$ denote the maximizers for $y$.  Then for all $y' \in X$,
\begin{align*}
\bar\psi(y')
& = \max_{x \in X, s \in S} \phi(x) - s \cdot (x-y')\\
& \geq \phi(x_y)- s_y \cdot (x_y-y')\\
& =  \phi(x_y)- s_y\cdot (x_y-y) - s_y \cdot (y-y')\\
& =  \bar\psi(y) - s_y \cdot (y-y')\\
& =  \bar\psi(y) + s_y \cdot (y'-y)
\end{align*}
Thus, $s_y$ is indeed a subgradient of $\bar\psi$ at $y$.

Since $\phi$ and $\bar\psi$ are feasible for the left hand side of \eqref{eqn:21}, we know that
$\phi(x) \leq \ell_S(x,y) +\bar\psi(y)$ for all $y \in X$.  Thus, we can define
$\bar\phi(x) = \inf_y \ell_S(x,y) + \bar\psi(y)$ and observe that $\bar\phi(x) \geq \phi(x)$.  Further, 
$\int_X \phi d\mu_+ - \int_X \psi d\mu_- \leq \int_X \bar\phi d\mu_+ - \int_X \bar\psi d\mu_-$.

We claim that $\bar\phi = \bar\psi$.  For one direction,
$\bar\phi(x) \leq \ell_S(x,x) + \bar\psi(x) = \bar\psi(x)$.
For the other,
$$\bar\phi(x) = \inf_y \ell_S(x,y) + \bar\psi(y) = \bar\psi(x) + \inf_y \ell_S(x,y) + \bar\psi(y) - \bar\psi(x) \geq \bar\psi(x).$$
The last inequality follows because $\bar\psi \in U(X,S)$ implies $\bar\psi(x) - \bar\psi(y) \leq \ell_S(x,y)$.
Thus taking $u = \bar\psi (= \bar\phi)$ is feasible for the right hand side and the value of the right hand side is at least that of the left hand side.
\raf{There is some scoring rule / general truthfulness stuff going on in this construction, but I forgot to tell you about it (maybe you already saw it), and probably will again...}
\end{proof}

The proof of the theorem then results by combining the lemmas.

\begin{proof}[of Theorem~\ref{thm:strong}]
\begin{align*}
\inf_{\substack{\gamma \in Radon_+(X \times X)\\ \gamma_1 - \gamma_2 \succeq_{X,S} \mu}} \int_{X \times X} \ell_S(x,y)d\gamma(x,y)
& \geq \sup_{u \in U(X,S)} \int_X ud\mu\\
& = \sup_{\substack{\phi(x)-\psi(y) \leq \ell_S(x,y)\\\phi,\psi \in \clo(\conv(U^\circ(X,S)))}} \int_X \phi d\mu_+ - \int_X \psi d\mu_-\\
& = \min_{\substack{\gamma \in Radon_+(X \times X)\\ \gamma_1 \succeq_{X,S} \mu_+\\ \mu_- \succeq_{X,S} \gamma_2}} \int_{X \times X} \ell_S(x,y)d\gamma(x,y)
\end{align*}
A feasible $\gamma$ for the final $\min$ is also feasible for the initial $\inf$, so the inequality is in fact an equality and the initial $\inf$ is a $\min$.
\end{proof}

\section{Concrete Examples}

In this section, we derive optimal mechanisms for a number of settings where previous duality results do not apply.
Since an optimal mechanism exists for all these settings, Theorem~\ref{thm:strong} guarantees that there exist a $u^*$ and $\gamma^*$ that lead to equality in equation~\eqref{eqn:strong}.  In the following examples, however, we need only apply the framework to find such a pair, as then Lemma~\ref{lem:weak} will certify optimality of $u^*$.  To aid us in finding such a pair we make use of the following powerful corollary of Lemma~\ref{lem:weak} that is a direct generalization from \cite{\third}, who describe it as providing complementary slackness conditions.

\begin{corollary}
\label{cor:structure}
Let $u^*$ and $\gamma^*$ be optimal in \eqref{eqn:strong}.  Then:
\begin{enumerate}
\item $\int_X u^* d\mu = \int_X u^* d(\gamma_1^* - \gamma_2^*)$ and 
\item $\ell_S(x,y) = u^*(x) - u^*(y)$ $\gamma^*(x,y)$-almost everywhere.
\end{enumerate}
\end{corollary}
\begin{proof}
By strong duality, the two inequalities in the proof of Lemma~\ref{lem:weak} are equalities for $u^*$ and $\gamma^*$.  This gives the first claim directly, and the second follows because $u^*(x) - u^*(y) \leq \ell_S(x,y)$ for all $x$ and $y$.
\end{proof}

From this corollary, and some basic facts about subgradients from convex analysis, we can identify some additional structure linking $u^*$ and $\gamma^*$ that will guide our search.
Let $\exp(S) = \{s \in S ~|~ \exists x,y \in X \text{ s.t. } s\cdot(x-y) = \ell_S(x,y) \}$ denote the union of exposed faces of $S$ with respect to $X$.\footnote{Equivalently, the set of points in $S$ supported by (dual) points in $C = X-X$.}  The following lemma shows that, except for a set of measure zero, these are the only permissible subgradients of $u^*$.

\begin{lemma}
\label{lem:boundary}
Points $x$ and $y$ share a subgradient of $u^*$ $s \in \exp(S)$ $\gamma^*(x,y)$-almost everywhere.  Furthermore, $(x-y)$ is a direction in which $s$ is exposed.
\end{lemma}
\begin{proof}
By Corollary~\ref{cor:structure}, $\ell_S(x,y) = u^*(x) - u^*(y)$ $\gamma^*(x,y)$-almost everywhere, so consider some $x$ and $y$ that are.  For $\lambda \in [0,1]$, let $f(\lambda) = u^*(\lambda x + (1-\lambda) y)$.  If $f$ is not affine, then $\ell_S(x,y) \neq u^*(x) - u^*(y)$, because the steepest slope present could have been used the whole way.  Thus, $f$ is affine and there is some $s \in S$ such that $s \in \partial u^*(y)$ (the set of subgradients of $u^*$ at $y$) and $s \cdot (x-y) = u^*(x) - u^*(y)$.  It is standard (see, e.g.,\cite[Lemma 2]{frongillo2014general}) that this implies $s \in \partial u^*(x)$.
Since  $\ell_S(x,y) = u^*(x) - u^*(y)$, $s \in \exp(S)$, and furthermore is maximal in the $(x-y)$ direction because that is a direction that shows it is in $\exp(S)$.
\end{proof}

Lemma~\ref{lem:boundary} significantly restricts the space of both optimal mechanisms $u^*$ and optimal duals $\gamma^*$.  For $u^*$, the only allocations are elements of $\exp(S)$.  For $\gamma^*$, when matching $x$ and $y$ which share allocation $s$, they much be chosen to be in a direction where $s$ is exposed.  For example, if $S = [0,1]^2$, when the allocation is $(1,0.5)$, $\gamma^*$ only matches pairs $x$ and $y$ where $x$ is to the right of $y$, while when the allocation is $(1,1)$ it is allowed to have $x$ be anything up and to the right of $y$.

\subsection{A Recipe for the Examples}

In each of the examples below, we follow the same basic recipe.  First, we state the measure $\mu$ for which the optimal mechanism maximizes $\int u^* d\mu$.
As discussed in Section~\ref{sec:objective}, determining $\mu$ is a direct application of the machinery of~\citet{\third}, so we simply give the measure without further comment.
For some examples, we then identify an alternate measure $\mu'$ and show that $\mu' \succeq_{X,S} \mu$.
The $\gamma^*$ we ultimately construct will have the property that $\gamma_1^* - \gamma_2^* = \mu$ or $\mu'$ as appropriate.

We then propose a mechanism $u^*$, which we will use to guide our construction of a matching $\gamma^*$ for which strong duality holds, certifying its optimality.  A useful construct in this regard will be the \emph{cell} of an allocation/subgradient $s\in S$, the set of types allocated $s$ under $u^*$, defined formally by $\cell(s) = \{x \in X ~|~ s \in \partial u^*(x)\}$, where $\partial u^*(x)$ denotes the set of subgradients of $u^*$ at $x$.
Lemma~\ref{lem:boundary} shows that $\gamma_1^*(\cell(s)) = \gamma_2^*(\cell(s)) = \gamma^*(\cell(s),\cell(s))$.  Thus, a necessary property of such a $u^*$ is that for each $s \in S$, $\mu(\cell(s)) = 0$, meaning that $\gamma^*$ essentially only matches types within the same cell.  This property, which we term the ``integrate to zero'' condition, guides our construction of $\gamma^*$ by allowing us to construct it on each cell independently.

As Corollary~\ref{cor:structure} tells us that $u^*(x)-u^*(y) = \ell_S(x,y)$ for almost every matched pair $(x,y)$, we must have $u^*(x)-u^*(y) = s \cdot (x-y)$ within a given $\cell(s)$; this lets us focus on pairs such that $s \in \argmax_{s'\in S} s'\cdot (x-y)$.  A matching which satisfies this cell condition will be said to ``respect $\ell_S$''.  Once we verify that $\mu(\cell(s)) = 0$ for $u^*$, we then show that such a $\gamma^*$ respecting $\ell_S$ exists, either by directly constructing it or by applying Strassen's~\citeyear{strassen1965existence} Theorem (or more precisely a special case of it~\citep{kamae1977stochastic,lindvall1999strassens}).

\begin{theorem}[Strassen's Theorem in 1 Dimension]
Let $P$ be a set of points on a line segment (and so totally ordered).   Let $\nu_+$ and $\nu_-$ be two unsigned measures on $P$ with $\nu_+(P) = \nu_-(P)$ such that $\nu_+ \succeq_{st} \nu_-$, where $\succeq_{st}$ denotes stochastic dominance.  Then there exists an unsigned measure $\gamma$ on $P \times P$ with marginals $\gamma_1 = \nu_+$ and $\gamma_2 = \nu_-$  such that $x > y$ $\gamma$-almost surely.
\end{theorem}

Strassen's Theorem allows us to prove the existence of the desired $\gamma^*$ without needing to construct it directly.  In many of the examples that follow, the mass of $\mu^+$ is restricted to a 1-dimensional set $P$, but $\mu_-$ need not be.
Nonetheless, this one dimensional set has some total order $\geq_P$ such that if $x_1 \geq_P x_2$ then for all $y$ in the cell, if $u^*(x_1) - u^*(y) = \ell_S(x_1,y)$ then $u^*(x_1,y) = \ell_S(x_1,y)$.  Intuitively, this ordering means that if $\gamma^*$ is ``permitted'' to match $x_2$ and $y$ in this sense, it is also permitted to match $x_1$ and $y$.  Thus, $\geq_P$ captures which points in $P$ are least restricted in terms of the matching, effectively reducing to the 1-dimensional case.  This reduction is captured by the following lemma.

\begin{lemma}
\label{lem:1dmatch}
Let $X' \subset X$, e.g. $X'=\cell(s)$, and $\nu_+$ and $\nu_-$ be two unsigned measures on $X'$ with $\nu_+(X') = \nu_-(X')$.  For $x\in X'$ define $A(x) = \{y \in X' ~|~ u^*(x) - u^*(y) = \ell_S(x,y)\}$ to be the points in $X'$ which can be matched with $x$.  Let $P \subset X'$ with $\nu_+(P) = \nu_+(X')$, such that the subset relation on $A(x)$ induces a total order on $P$ defined by $x\geq_P x' \iff A(x') \subseteq A(x)$.  Letting $B(x) = \{ x' \in P ~|~ x \geq_P x'\}$, if $\nu_-(A(x)) \geq \nu_+(B(x))$ for all $x \in P$, there exists a measure $\gamma$ on $X' \times X'$ with marginals such that $\gamma_1 - \gamma_2 = \nu_+ - \nu_-$ and $u^*(x) - u^*(y) = \ell_S(x,y)$ $\gamma$-almost surely.
\end{lemma}

\begin{proof}
Define $\nu_-'$ an unsigned measure on $P$ by $\nu_-'(P') = \nu_-( \{y \in X' ~|~ \exists x \in P'. u^*(x) - u^*(y) = \ell_S(x,y)\}$.  Observe that, by construction, $\nu_-'(B(x)) = \nu_-(A(x))$ for all $x \in P$.  By assumption, $\nu_-'(B(x)) = \nu_-(A(x)) \geq \nu_+(B(x))$ for all $x \in P$, so $\nu_+ \succeq_{st} \nu_-'$.  Thus by Strassen's Theorem, there exists a $\gamma'$ on $P \times P$ with marginals $\gamma_1' = \nu_+$ and $\gamma_2' = \nu_-'$ and $x \geq_P x'$ $\gamma'$-almost surely.

Let $\gamma''(\{(x_i,y_i)\}) = \nu_-(\{y_i ~|~ x_i= \min \{x ~|~ y_i \in A(x)\}\})$ (where the $\min$ is taken according to $\geq_P$, also implicitly requiring that $x_i \in P$).
Note that $\gamma_1''$ is a measure on $P$ and $\nu_-' \succeq_{st} \gamma_1''$.  Applying Strassen's theorem again, we get a $\gamma'''$ with marginals $\gamma_1''' = \nu_-'$ and $\gamma_2''' = \gamma_1''$.
Observe that $P \subset X'$, so if $x \geq_P x'$ then $u^*(x) - u^*(x') = \ell_S(x,x')$.
Therefore, we can take $\gamma = \gamma' + \gamma'' +\gamma'''$, and have it satisfy $u^*(x) - u^*(y) = \ell_S(x,y)$ $\gamma$-almost surely.  Further, $\gamma_1 - \gamma_2 = \nu_+ + \gamma_1'' + \nu_-' - \nu_-' -\nu_- - \gamma_1'' = \nu_+ - \nu_-$.  
\end{proof}

Having constructed such a $\gamma^*$, we will then have shown $\int_X u^*d\mu = \int_{X\times X} \ell_S \,d\gamma^*$, which certifies optimality of $u^*$ by weak duality (Lemma~\ref{lem:weak}).

\subsection{Two items, allocate at most one}
\label{sec:example-at-most-one}

\begin{figure}
  \centering
  \begin{tikzpicture}[scale=3]
    \draw (0,0) -- (1,0) -- (1,1) -- (0,1) -- cycle;
    \fill (0,0) circle [radius=1pt];
    \node at (0.1,0.1) {+1};
    \node at (0.5,0.5) {-3};
    \node at (0.5,0.9) {+1};
    \node at (0.9,0.5) {+1};
    \draw[line width=2pt] (1,0) -- (1,1);
    \draw[line width=2pt] (0,1) -- (1,1);
  \end{tikzpicture}~~~
  \begin{tikzpicture}[scale=3]
    \draw (0,0) -- (1,0) -- (1,1) -- (0,1) -- cycle;
    \fill (0,0) circle [radius=1pt];
    \draw[line width=2pt] (1,0) -- (1,1);
    \draw[line width=2pt] (0,1) -- (1,1);
    \draw[dashed] (0,0.577) -- (0.577,0.577);
    \draw[dashed] (0.577,0) -- (0.577,0.577) -- (1,1);
    \node at (0.3,0.8) {$(0,1)$};
    \node at (0.3,0.3) {$(0,0)$};
    \node at (0.8,0.3) {$(1,0)$};
  \end{tikzpicture}~~~
  \begin{tikzpicture}[scale=3]
    \fill[blue!40] (0.577,0.577) -- (1,1) -- (1,0.6) -- (0.577,0.577-0.4) -- cycle;
    \draw (0,0) -- (1,0) -- (1,1) -- (0,1) -- cycle;
    \fill (0,0) circle [radius=1pt];
    \draw[line width=2pt] (1,0) -- (1,1);
    \draw[line width=2pt] (0,1) -- (1,1);
    \draw[dashed] (0,0.577) -- (0.577,0.577);
    \draw[dashed] (0.577,0) -- (0.577,0.577) -- (1,1);
    \draw[blue,line width=2pt] (1,0.6) -- (1,1);
    \def\thexval{1.2}
    \draw (\thexval-0.05,1) -- (\thexval+0.05,1);
    \draw (\thexval,1) -- (\thexval,0.6);
    \draw (\thexval-0.05,0.6) -- (\thexval+0.05,0.6);
    \node at (\thexval+0.08,0.8) {$a$};
    \node at (1.06,0.8) {$B$};
    \node at (0.8,0.62) {$A$};
  \end{tikzpicture}
  \caption{(L) The $\mu$ for the two-item uniform valuaton setting.  (M) The optimal mechanism in this setting when at most one item can be sold.  (R) Construction of a matching $\gamma^*$ certifying optimality.}
  \label{fig:two-items-at-most-one}
\end{figure}
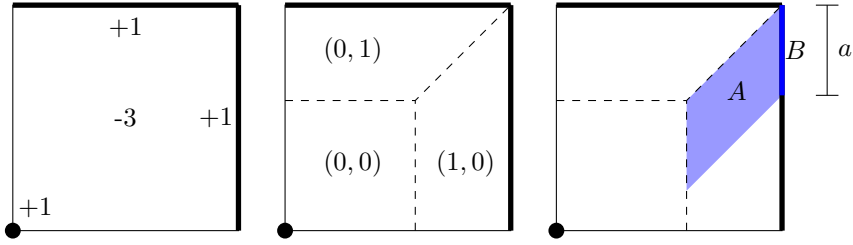

In this setting, values are uniform $[0,1]$ for each item and $S = \{(p_1,p_2) ~|~ 0 \leq p_1,p_2 \wedge p_1 + p_2 \leq 1\}$.  This gives $\mu$ with a mass of $-3$ uniformly distributed on the interior of the unit square, a point mass of $+1$ at the origin, and a mass of $+1$ uniformly distributed on each of the lines $(0,1)-(1,1)$ and $(1,0)-(1,1)$.  (See Figure~\ref{fig:two-items-at-most-one}.)

As we will show, the optimal mechanism here is to offer the choice of either item at a price of $1 / \sqrt{3}$.  This mechanism divides the space into 3 cells: $(0,0)$ when neither item worth more than $1/\sqrt{3}$, $(1,0)$ when at least one item is worth $1/\sqrt{3}$ and item one is preferred to item two, and $(0,1)$ when at least one is worth $1/\sqrt{3}$ with item two preferred.  Each cell satisfies the integrate to 0 condition: the no-sale cell has a negative mass inside of $\int_{0\leq p_1,p_2\leq 1 / \sqrt{3}} -3 dp_1 dp_2  = -3(1/\sqrt{3})^2 = -1$ and there is a point mass of 1 at the origin.  By symmetry, the interior of the other cells also integrate to $-1$, yielding $0$ when including the boundary.  Thus, we may turn to constructing $\gamma^*$ which only matches points within one of these cells.  For the $(0,0)$ cell, this matching is trivial, as we match all the negative mass to the point mass at $(0,0)$.

It remains to show that for the cell where (without loss of generality) item 1 is sold, we can match the interior to the boundary such that a point $(x_1,x_2)$ is matched to a boundary point $(1,x_2')$ with $1-x_2' \geq x_1 - x_2$, thus respecting $\ell_S$.  For this we apply Lemma~\ref{lem:1dmatch}, with $X'$ being the $(1,0)$ cell, $P$ being its right boundary, and of course $\nu = \mu$, as $\mu_+(\cell(1,0)) = \mu_-(\cell(1,0))$ by the above (we drop the double parentheses when plugging in coordinates).  Here $A(1,x_2') = \{x ~|~ 1-x_2' \geq x_1 - x_2\} \cap \cell(1,0)$ is the set of points which can be matched with boundary point $(1,x_2')$ and $B(1,x_2') = \{(1,x_2) ~|~ 1 \leq x_2 \leq x_2'\}$ is an upward-closed portion of the boundary; see Figure~\ref{fig:two-items-at-most-one}(R).  Thus, we need only check that for all $a\in[0,1]$, $\mu_+(B(1,1-a)) = a \leq \mu_-(A(1,1-a))$.
This reduces to two cases: (i) $3(1-1/\sqrt{3})a \geq a$ when $0 \leq a \leq 1/\sqrt{3}$, which holds for all $a\geq 0$, and (ii) $1-\frac 3 2 (1-a)^2 \geq a$ when $1\leq a \leq 1/\sqrt{3}$, which is easily checked on the two endpoints $a=1$ and $a=1/\sqrt{3}$.

\subsection{Two items, allocate exactly one}
\label{sec:example-exactly-one}

In this setting, values are uniform $[0,1]$ for each item and $S = \{(p_1,p_2) ~|~ 0 \leq p_1,p_2 \wedge p_1 + p_2 = 1\}$.  Thus $\mu$ is the same as in the previous example.  We claim that the optimal mechanism is to offer the choice of either the first item at a price of $1/3$ or the second item for free.  This mechanism divides the space into two cells, separated by the line $x_1 = x_2 + 1/3$.  Each cell satisfies the integrate to 0 condition: the $(1,0)$ cell has a mass of $2/3$ along the boundary and the interior mass is $\int_{1/3}^1\int_0^{x_1-1/3} -3 dx_2dx_1 = -2/3$.  It remains to show that for each cell, we can match the interior to the boundary such that $\ell_S(x,y) = y_1 - x_1$ in cell 1 and similarly for cell 2.  For the $(1,0)$ cell this is easily seen by the same sort of application of Lemma~\ref{lem:1dmatch} as before, this time looking at the point $(1,a)$ with $B(1,a)$ again being upward closed: $\tfrac 2 3 - \tfrac 3 2 a^2 \geq \tfrac 2 3 - a$.  For the $(0,1)$ cell, the matching is a bit more complex; Figure~\ref{fig:exact21}(L) diagrams the matching.  The dashed line again divides the cells, with the lower right being the $(1,0)$ cell, for which we have already described the matching.  The $(0,1)$ cell has several areas of positive mass to be matched: the top left, two in the top right, and two in the bottom left.  The top left we match the same as the $(1,0)$ cell.  For the rest, our matching respects $\ell_S$ if points on the boundary are matched with points along a line at a 45 degree angle from them (because it is a translation of $y=x$); we match the two regions in the top right in this way.  For the bottom left, we take advantage of the fact that we are allowed to do mean-preserving spreads: we may replace $\mu$ with some $\mu'$ where all we have done is spread (positive) mass out such that at the end of any spreading operation the center of mass of the mass that has been spread is unchanged (See \citep[Remark 2]{\third} for additional discussion of why this is permissible).  Here we take a mass on the line $y=x$ and spread it in the $(-1,1)$ and $(1,-1)$ directions evenly to cancel out the negative mass in the cell, leaving a residual negative mass of $-1$ on the line $y=x$.  With this $\mu'$ (which clearly has not changed the matching in the other cells) we match this line of negative mass with the origin.

\begin{figure}
  \centering
  \begin{tikzpicture}[scale=3]
    \draw (0,0) -- (1,0) -- (1,1) -- (0,1) -- cycle;
    \fill (0,0) circle [radius=1pt];
    \draw[line width=2pt] (1,0) -- (1,1);
    \draw[line width=2pt] (0,1) -- (1,1);

    \fill[blue!40] (0.33,0) -- (1,0.66) -- (1,0.4) -- (0.33+0.66-0.4,0) -- cycle;
    \draw[blue,line width=2pt] (1,0.4) -- (1,0.66);
    \def\thexval{1.2}
    \draw (\thexval-0.05,0.4) -- (\thexval+0.05,0.4);
    \draw (\thexval,0.4) -- (\thexval,0);
    \draw (\thexval-0.05,0) -- (\thexval+0.05,0);
    \node at (\thexval+0.08,0.2) {$a$};
    \node at (1.06,0.5) {$B$};
    \node at (0.7,0.27) {$A$};

    \draw[dashed] (0.33,0) -- (1,0.66);
    \draw[dotted] (0,0.33) -- (0.66,1);
    \draw[dotted] (0.33,0.66) -- (0.66,0.66);
    \draw[dotted] (0.66,0.33) -- (0.66,0.66);
    \draw[dotted] (0,0) -- (1,1);
    \draw[dotted] (0,0.33) -- (0.66,1);
    \draw[>=latex,->] (0.55,0.75) --++ (0.2,0.2);
    \draw[>=latex,->] (0.75,0.55) --++ (0.2,0.2);

    \foreach \x in {2,...,5} {
    \draw[>=latex,->] (0.1*\x,0.1*\x) --++ (0.1,-0.1);
    \draw[>=latex,->] (0.1*\x,0.1*\x) --++ (-0.1,0.1); }
  \end{tikzpicture}
  ~~~
  \begin{tikzpicture}[xscale=2,yscale=2]
    \draw[>=latex,->] (0,0) -- (2.6,0);
    \draw[>=latex,->] (0,0) -- (0,1.5);

    \draw[dashed] (0,1.23) -- (0.93,0.3) -- (2.6,0.3);
    \draw[dashed] (0.93,0.3) -- (0.93,0);
    \draw[red,dotted] ($(0,2)+0.3*(1,-2)$) edge (1,0);
    \node at (0.97,0.07) {-};
    \draw[red,dotted] (2.28,0.3) edge node[right] {+} (2.36,0.14);
    \draw[>=latex,->] (1.15,0.55) --++ (0.3,0.3);
    \draw[>=latex,->] (1.6,0.15) --++ (0.4,0);
  \end{tikzpicture}
  ~~~~
  \begin{tikzpicture}[scale=3]
    \def\offset{0.3}
    \fill (\offset,1-\offset) circle [radius=.7pt];
    \fill (1-\offset,\offset) circle [radius=.7pt];
    \draw[>=latex,->] (\offset,1-\offset) --++ (0,\offset);
    \draw[>=latex,->] (1-\offset,\offset) --++ (0,-\offset);
    \foreach \x in {0,0.5,1} {
      \draw[black!80] (\x,0) -- (\x,1);
      \draw[black!80] (0.1,0.6*\x+0.2) -- (0.9,0.6*\x+0.2);
    }
    \foreach \x in {0.8,0.9,1} {
      \draw[black!80] (\x-0.8,\x) -- (\x,\x-0.8);
    }
    \draw[red!80] (0.35,1) -- (0.65,0);
  \end{tikzpicture}
  \caption{(L) Matching for two items, allocating exactly one. (M) Matching for the deterministic mechanism with two items with exponentially-distributed valuations.
 (R) Depiction of a non-mean-preserving operation which is still permitted in the deterministic setting.}
  \label{fig:exact21}
  \label{fig:spread}
  \label{fig:exponential21}
\end{figure}
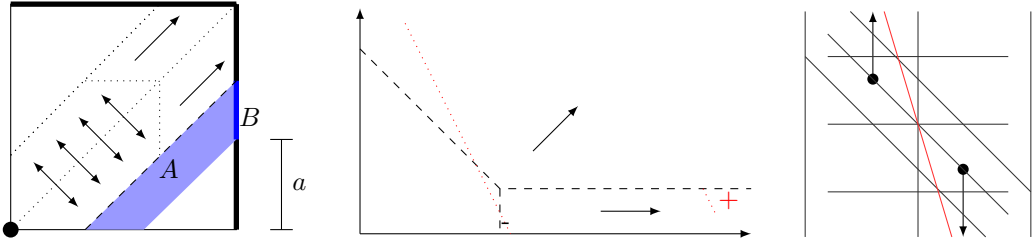

\subsection{Deterministic allocations with two exponential items}
\label{sec:example-deterministic}

In this setting, values are exponentially distributed on $[0,\infty)$ with $\lambda_1 = 2$ and $\lambda_2 = 1$.  Section 9 of~\citet{\first} analyzes this setting and shows that the optimal mechanism is randomized (sell both items at a price $p^*$ of approximately 1.2319, sell item 1 with probability 1 and item 2 with probability 0.5 for a price of 1, sell neither for a price of 0).  Here we consider the optimal deterministic mechanism, captured by the restriction $S = \{(0,0),(1,0),(0,1),(1,1)\}$, which turns out to be the following: sell both for a price of approximately 1.2286, and sell just item 1 for a price of approximately 0.9288.%
\footnote{The exact values are those which cause the integrate to zero condition to be satisfied in all three cells.}
This case is a good example to illustrate how our approach extends to deterministic mechanisms (i.e. non-convex sets of subgradients), since it is the perhaps then simplest case where the optimal randomized mechanism is not in fact deterministic.

To verify the optimality of the deterministic mechanism, we need to match each of the regions shown in Figure~\ref{fig:exponential21}(M).  The zero region integrates to 0, so every point can be matched with the origin.  However, the other two regions are more problematic, and cannot immediately be matched.  Consider the cell where item 1 is sold.  To get a ``tight'' matching, we can only match along horizontal lines as otherwise we get the added cost from the $(1,1)$ subgradient.  This leaves negative mass we cannot match at the left boundary with positive mass out toward infinity, as shown in red in Figure~\ref{fig:exponential21}(M), where the main dashed red line has the property that integrating from it to infinity yields zero.  This same line shows that there is also a problem in the cell where both goods are sold, because the line passes through the zero cell.  This means that there is excess positive mass in the lower right ``corner'' of this cell which we cannot easily match.

The key to resolving this is shown in Figure~\ref{fig:spread}(R).  Suppose we have positive mass at the two black points.  We claim that we can replace $\mu$ with a measure $\mu'$ given by moving mass up from the first one and down from the second one in equal amounts, in a non-mean-preserving spread.  To see why, consider where the two points sit relative to any possible cell boundary.  Since there are only 3 allocations, there are only 9 cases (the gray lines).  In 6 of them, the two points end up in the same cell (at least relative to that boundary), so this move does not change the integral of the function because if we do this to two points in the same cell it is effectively a mean-preserving spread.  In the other three cases we can only increase the value:
\begin{enumerate}
\item When the boundary is vertical this move has no effect because a vertical boundary means only item 1 is allocated and we are only changing values for item 2.
\item When the boundary is horizontal this increases the value because we are increasing the value for item 2 in a cell where item 2 is allocated and decreasing it where it is not (or is allocated to a lesser extent).
\item When the boundary is at 45 degrees the same applies, just item 1 happens to be (irrelevantly) allocated as well.
\end{enumerate}
While we have drawn this with the line between the points at 45 degrees, but all that matters is that it is that or steeper.

Note that since we need a strictly lower optimal dual value than the randomized case, for this operation to be helpful, it must not be permissible in the randomized case.  Indeed this is the case, as we illustrate by the red line, which is a possible boundary in the randomized case: with such a boundary, our operation decreases the value of the integral, so it would not be permitted.

We also need to find some mass to move up, and moving mass in the item 1 cell does not help since we need to move all that positive mass downward.  The key is to observe that 45 degree lines near the left end of the item 1 cell pass into the zero cell, and that instead of matching the negative mass there to the positive mass at the origin we can move the mass at the origin to those points (which is equally costless).  The positive mass is still too far out in the $(1,0)$ direction, but we can fix this with a mean-preserving spread along each horizontal line where we move the mass we need towards zero in exchange for pushing some mass even further out (this has no effect on the cost of the matching).

To be more explicit about the construction, let $pos1(z) = \int_{x_1= 0.9288}^\infty \mu(x_1,1.2319 - 0.9288-z)dx_1$ be the excess positive mass along the horizontal line a distance $z$ down in the item 1 cell.  Similarly, let $neg1(z) = \int_{x_2=0}^{2(1- 0.9288-z)} \mu( 0.9288+z,x_2)dx_2$ be the excess negative mass below the dashed red line in Figure~\ref{fig:exponential21}(M) along the vertical line at $x_1 = 0.9288+z$.  It is easy to verify that these new measures have the same total measure, are 1 dimensional, and $pos(z)$ stochastically dominates $neg(z)$.  Thus, we can find a matching such that that the excess positive mass from $z_1$ is matched to excess negative mass $z_2$ such that $z_2 < z_1$.  We cannot give a closed form for this, but as we will see this is not needed.

To move this mass downward as described (after first performing a mean-preserving spread so the mass we need to move downward is just a point mass), we need to move mass in the zero cell upward.  Conveniently, we can do so by moving equal mass on the origin to any point on the 45 degree line up and to the left from this point mass.  Thus, it suffices to verify that sufficient negative mass always exists on such lines.  In fact, let $(x_1^*,x_2^*)$ be the point at which the red dashed line in Figure~\ref{fig:exponential21}(M) crosses the boundary between the cell where no item is sold and the cell where both items are sold.  We use only the mass at points where $x_1 < x_1^*$, as we will use the mass at points with $x_1 \geq x_1^*$ for \raf{to?} \ian{I do mean for} the cell where both items are sold.  It is straightforward to verify numerically that such a line that has the least negative mass (which is exactly the boundary between the zero cell and the cell where both items are sold), has more mass than any value of $neg(z)$ by at least a factor of 20.

Having created this $\mu'$, we have $\int_{x_1 = 0.9288}^\infty d\mu'(x_1,x_2) = 0$ for all $x_2 < 1.2319 - 0.9288$.  On each of these lines the negative mass is entirely to the left of the positive mass, so by a direct application Strassen's Theorem we can construct a $\gamma^*$ that matches along each of these lines.

Similarly, in the cell where both items are sold, we can move mass at a 45 degree angle up and to the left while moving mass in the zero cell at a 45 degree angle down and to the right.   Once this is done, we can then match this cell in the standard way (each negative point matched to positive point(s) up and to the right of it) via an application of Theorem 7.4 of~\citet{\first} (another application of Strassen's theorem), in the same manner as their analysis in their Theorem 9.3.

To make this more precise, let $pos2(z) = \int_{x_1= 0.9288-z}^\infty \mu(x_1,1.2319 - 0.9288+z)dx_1$ be the excess positive mass along the horizontal line a distance $z$ up in the cell where both items are sold.  Similarly, let $neg2(z) = \int_{x_1=0}^{x_1^*-z/2} \mu(x_1,x_1^*+x_2^*+z/2-x_1)dx_1$ be the excess negative mass below on the 45 degree line inside the dashed red line shifted upward from the cell boundary by $z$.  Again Strassen's theorem shows we can find an appropriate place to shuffle the excess positive mass to the desired negative mass.  To counteract this moving of positive mass up and to the left on a 45 degree angle, we move mass from the origin to some point to the right of the point we are moving the excess positive mass out of to create a negative point mass.  By construction, these are all at or to the left of $x_1^*$, thus it suffices to verify that the total mass to the right of $x_1^*$ is the zero cell excess the total mass we need to move, which is easily done numerically. 
Note that the shuffle we performed in this cell had no effect on the ultimate cost of this matching because moving mass at a 45 degree angle up and to the left exactly counters the increase in cost from the leftward movement with the decrease in cost from the upward movement.

Having created this $\mu'$, we can now apply Theorem 7.4 of~\citet{\first}.  By construction, we can first match all the horizontal lines with $x_2 \leq x_2^*$ in the same manner we did in the cell where item 1 is sold.  For the remainder of the cell the application of  Theorem 7.4 of~\citet{\first} is now straightforward because we have because as in their application of it the zero cell is entirely below the dashed red line (what they term the ``absorbing hyperplane''), and our construction of $\mu'$ has not affected any mass above or to the right of the dashed red line.

\subsection{Optimal mechanism for two non-additive items}
\label{sec:example-nonadd}

Consider a setting values are uniform $[0,1]$ for each item individually but getting both items is twice as valuable as their sum.  Thus $X = \{(x1,x2,x3) ~|~ 0\leq x_1 \leq 1, 0 \leq x_2 \leq 1, x_3 = 2 (x_1 + x_2)\}$ with $S = \{ (p_1,p_2,p_3) ~|~ 0 \leq p_1,p_2,p_3 \wedge p_1+p_2+p_3 \leq 1\}$.  However, it is convenient to instead use the equivalent representation $X = \{(x1,x2) ~|~ 0\leq x_1,x_2 \leq 1\}$ with $S=\conv(\{(0,0),(1,0),(0,1),(2,2)\})$, which allows us to use the same $\mu$ as in our first two examples.

We show that grand bundling at a price of $\sqrt{8/3}$ is optimal.  (Note that this choice causes both cells to satisfy the integrate to zero condition).
To construct an appropriate $\gamma^*$, observe that a point $(x_1,x_2)$ with $x_1 > x_2$ and $2(x_1 + x_2) > \sqrt{8/3}$ can be matched with a point $(1,x)$ if $(2,2) \cdot ((1,x)-(x_1,x_2)) \geq (1,0) \cdot ((1,x)-(x_1,x_2))$.  Again we apply Lemma~\ref{lem:1dmatch} to show that one half of the cell can be matched to the right boundary and the other, by symmetry to the top boundary.  In particular, we take $A(1,x_2') = \{x ~|~ x_1 > x_2 \wedge 2(x_1 + x_2) \geq \sqrt{8/3} \wedge 2(1-x_1 +x_2' - x_2) \geq 1 - x_1\} $ is the set of points in the correct half of the cell which can be matched with boundary point $(1,x_2')$ and $B(1,x_2') = \{(1,x_2) ~|~ 0 \leq x_2 \leq x_2'\}$ is a downward-closed portion of the boundary.
Thus, we need only check that for all $a\in[0,1]$, $\mu_+(B(1,1-a)) = a \leq \mu_-(A(1,1-a))$, which is easily verified numerically.

In fact we can generalize this analysis to the case where bundling is $\alpha$ times as valuable and conclude that grand bundling is optimal at a price of $\alpha\sqrt{2/3}$ if and only if $\alpha$ is greater than approximately $1.24$.

\section{Discussion}

We have given a duality framework that generalizes that of \citet{\third} by allowing restrictions on the set of allowable allocations.  This allows us to model 
problems that previous frameworks could not handle, such as
\begin{itemize}
\item non-additive valuations,
\item multiple copies of items,
\item a requirement that the mechanism be deterministic, and
\item allocation restrictions such as matroid constraints.
\end{itemize}
One particularly interesting combination is  an arbitrary convex type space with the constraint that the allocation must be a probability distribution.  This recovers mechanism design.  Thus, in a sense our framework is fully general.

We have provided both a weak duality result (Lemma~\ref{lem:weak}) that provides a recipe for certifying the optimality of mechanisms, as well as a strong duality result (Theorem~\ref{thm:strong}) which guarantees that such a certificate always exists, even in settings where an optimal mechanism does not exist.
The key technical tool we use is the observation that we can restrict the set of convex functions to those that correspond to mechanisms that satisfy our allocation restriction, leading to an alternative integral stochastic order $\succeq_{X,S}$.  Using this, as well as facts from convex analysis about subgradients, we provide a new characterization result (Lemma~\ref{lem:boundary}) that identifies the set of allocations an optimal mechanism can use as those that are exposed points of the set.

For the remainder of the paper, we turn to additional intuition from this lemma, and then discuss some of the limitations of our results, how these can be relaxed, and other future research directions.

\subsection{The Simple Economics of Optimal Mechanism Design}

In their classic paper, \citet{bulow1989simple} reinterpreted Myerson's~\citeyear{myerson1981optimal} characterization of optimal single item auctions in terms of pricing.  In particular, they observed that his characterization in terms of virtual valuations could instead be described as finding a price such that the derivative of expected revenue with respect to that price is zero.
In the examples, we made frequent use of the necessary condition that $\mu(\cell(s)) = 0$ for each allocation $s \in S$.  As we show, this necessary condition has a simple economic interpretation: the partial derivative of expected revenue with respect to that price must be zero in an optimal auction.

Consider an allocation $s$ that is purchased with positive probability in the optimal mechanism given by $u^*$.
$$\int_{\cell(s)} u^* d\mu = \int_{\cell(s)} (x \cdot s - p) d\mu(x).$$
The marginal revenue from selling $s$ is the derivative of this with respect to $p$, which is $\int_{\cell(s)} -d\mu$, and by the envelope theorem, this is the marginal revenue of the overall mechanism.
Thus, by the optimality of $u^*$ we must have $\mu(\cell(s)) = 0$.

This analysis also shows why this necessary condition is not sufficient in our multidimensional setting: for all fixed choices of allocations to offer, there exists an optimal pricing, but that does not mean that the menu we have chosen to price is the optimal menu.

\subsection{Zero Set Mechanisms}

Daskalakis et al. \citeyear{\third} define a class of mechanisms for the additive setting they call zero set mechanisms.  Such mechanisms have two features.  First, they either allocate nothing or allocate at least one good with probability 1.  Second, they are entirely determined by the boundary of the cell where no goods are allocated (i.e. the zero set).  Geometrically, this means the cell for every other allocation borders the zero set.  All known optimal mechanisms for independently distributed items are in this class.

Lemma~\ref{lem:boundary} shows that the first property is not an accident.  The only allocations in $\exp(S)$ that are not permissible in a zero set mechanism are those which allocate at least one item with probability zero and there is no item allocated with probability one.  Such allocations would only be permissible if $\gamma^*$ were such that we could match all the $x$s with that allocations to points $y$ directly above them.  All explored $\mu$ render this impossible and it is possible that in fact no such $\mu$ exists and all optimal mechanisms for the additive setting satisfy this first property of zero-set mechanisms.

\subsection{Approximate Optimality of Deterministic Mechanisms}

A significant current research direction is identifying simple classes of mechanism that are approximately optimal in some setting~\citep{hart2012approximate,alaei2013simple,giannakopoulos2015bounding}.
One such simple class is deterministic mechanisms.  Let $D$ be a finite set of deterministic allocations and let $\conv(D)$ be the set of allocations that randomizes over them.  Note that $\ell_D = \ell_{\conv(D)}$, which means that the dual optimization problems differ only in their set of feasible $\gamma$.  An interesting future research direction is to try and use this fact to give a bound on the ratio between optimal revenue and optimal deterministic revenue settings of interest.

\subsection{Limitations}

As discussed, our framework is fully general, in that it encompasses all of mechanism design.  However, there are a number of technical caveats that still restrict the applicability of our results.  Most notable is our reliance on Theorem~\ref{thm:mu}, which imposes significant restrictions on the prior $f$ over agent types.  Most notably it imposes a differentiability requirement, as well as requiring the support of $f$ to be a subset of a compact set.  Similarly, several steps of the proof of Theorem~\ref{thm:strong} rely on the compactness of $X$.  We know that these restrictions can be relaxed at least somewhat, as the example of exponentially distributed $f$ shows.  \citet{\third} has some discussion related to this relaxation, but a fuller understanding is an important direction in understanding the limitations of this approach.

A smaller restriction is our requirement that $S$ be closed.  This is primarily used in Lemma~\ref{lem:21}.  It is unclear how important relaxing this restriction his however, as most natural problems seem to yield a closed $S$.  Further, cases with open $S$ are often ones where no optimal mechanism exists, although it is possible to construct contrived examples where one still exists.%
\footnote{Take any closed $S$ and add an additional outcomes, but restrict $X$ so that the value for that outcomes is always zero.  Extend $S$ so that it is open only with respect to that outcome and an optimal mechanism will still exist because the allocation for that outcome is irrelevant to agent utility.}

Finally, while not a limitation of our framework per se, applying it to practical examples requires constructing a $\gamma^*$ to certify the optimality of $u^*$.  The primary tool we have for this is Strassen's theorem, but this is a difficult tool to apply in more than two dimensions.  \citet{\second}\citeyear{giannakopoulos2015selling} have had some success in constructing matchings by using Hall's theorem on finite approximations of the problem, but even this approach is difficult to apply in higher dimensions.  Better tools here would greatly extend the practical reach of our framework.

\bibliography{diss,optimal}
\end{document}